\documentclass[12pt,reqno]{amsart}
\usepackage[hypertex]{hyperref}
\usepackage[final]{graphicx}
\usepackage{amsfonts}
\usepackage{amsmath}
\usepackage{amssymb}
\usepackage{amsthm}

\newtheorem{theorem}{Theorem}[section]

\newtheorem{lemma}[theorem]{Lemma}

\newtheorem{remarks}[theorem]{Remarks}

\newcommand{\R}{{\mathord{\mathbb R}}}

\begin{document}
\title[Indirect Coulomb Energy]{A new estimate on the indirect Coulomb Energy}

\author[Benguria]{Rafael D. Benguria$^1$}

\thanks{R.\ B.\ was supported in part by
the Iniciativa Cient'fica Milenio, ICM (CHILE) project P07--027-F. }

\author[Bley]{Gonzalo Bley$^2$}

\thanks{G.\ B.\ was supported in part by
the Iniciativa Cient'fica Milenio, ICM (CHILE) project P07--027-F. }

\author[Loss]{Michael Loss$^3$}

\thanks{M.\ L.\ was supported in part by NSF grant DMS-0901304.}

\address{$^1$ Departmento de F\'\i sica, P. Universidad Cat\'olica de Chile, Casilla 306, Santiago 22, Chile,}
\email{\href{mailto: rbenguri@fis.puc.cl}{ rbenguri@fis.puc.cl}}

\address{$^2$ Departmento de F\'\i sica, P. Universidad Cat\'olica de Chile, Casilla 306, Santiago 22, Chile,}
\email{\href{mailto: gabley@uc.cl}{ gabley@uc.cl}}

\address{$^3$ Georgia Institute of Technology, School of Mathematics,
Atlanta, Georgia 30332-0160, }
\email{\href{mailto:loss@math.gatech.edu}{loss@math.gatech.edu}}

\maketitle

\begin{abstract}
Here we prove  a new lower bound on the indirect Coulomb energy in quantum mechanics, in terms of the single particle density of the system. The new universal lower bound is an alternative to the  classical Lieb--Oxford bound (with a smaller constant, $C=1.45 < C_{LO}=1.68$) but involving an additive kinetic energy term  of the single particle density as well. 
\end{abstract}

\medskip {\sl .}
\date{today}

\section{Introduction}
Consider a system of $N$ particles  with charges $e_1, \ldots, e_N \geq 0$. In non relativistic quantum mechanics, this system is described by a normalized wavefunction,
\begin{equation}
\psi(x_1, \ldots, x_N; \sigma_1, \ldots, \sigma_N),
\label{eq:e1}
\end{equation}
where $x_1, \ldots, x_N$ denote the coordinates of the particles ($x_i \in \mathbb{R}^3$), and $\sigma_i$ 
represent possible discrete quantum numbers, such as spin. The corresponding probability density function is given by,
\begin{equation}
f(x_1, \ldots, x_N) = \sum_{\sigma_1, \ldots, \sigma_N} \left|\psi(x_1, \ldots, x_N; \sigma_1, \ldots, \sigma_N)\right|^2.
\label{eq:e2}
\end{equation}
We define the charge density of particle $i$ as,
\begin{equation}
\rho_i(x) = e_i\int f(x_1, \ldots, x_{i - 1}, x, x_{i + 1}, \ldots, x_N) \widehat{{dx}_i},
\label{eq:e3}
\end{equation}
where $\widehat{{dx}_i}$ means integration in all the particle coordinates, except the $i$-th. 
Finally, we then define the single particle 
density as,
\begin{equation}
\rho(x) = \sum_{i = 1}^N \rho_i (x),
\label{eq:e4}
\end{equation}
which is a charge density for the whole system. Here, we assume a Coulomb interaction between the particles given as usual by
\begin{equation}
\sum_{1 \le i < j \le N} \frac{e_i \, e_j}{|x_i-x_j|}.
\label{eq:e5}
\end{equation}
The expectation value of this interaction when the system is described by the wavefunction $\psi$ can be 
simply expressed as,
\begin{equation}
\sum_{i < j}\int \frac{f(X)}{\left|x_i - x_j\right|} dX,
\label{eq:e5b}
\end{equation}
where $X = (x_1, \ldots, x_N)$ and $dX = dx_1 \ldots dx_N$. This expression can be decomposed into 
two parts,
\begin{equation}
\frac{1}{2} \int\frac{\rho(x)\rho(y)}{\left|x - y\right|} \, dx \, dy + E,
\label{eq:e6}
\end{equation}
where the first term, which it is usually called the direct part of the Coulomb energy, represents a classical expression for the electrostatic energy of a charge 
distribution $\rho$. On the other hand, the remainder $E$ is all that has been missed by treating the system as such a distribution $\rho$, and it is known as the indirect part. In this article we will be interested in finding a lower bound for $E$.

\bigskip
In 1930, Dirac \cite{Di30} gave the first approximation for the indirect Coulomb energy in terms of the single particle density. In the standard case, when $e_i=e$, the absolute value of the charge of the electron (for all $i$), using an argument with  plane waves, he approximated $E$ by 
\begin{equation}
E \approx -c_D e^{2/3} \int \rho^{4/3} \, dx,
\label{eq:dirac}
\end{equation}
where $c_D=(3/4)(3/\pi)^{1/3} \approx  0.7386$ (see, e.g., \cite{Mo06}, p. 299).

 The first rigorous lower bound for $E$ was obtained by E.H. Lieb in 1979 \cite{Li79}, using the Hardy--Littlewood Maximal Function \cite{StWe71}. 
There he found that,
\begin{equation}
E \geq -8.52\left\{\int\left[\sum_{i = 1}^N e_i^{2/3}\rho_i(x)\right]^{4/3} dx\right\}^{3/4}\left[\int \rho(x)^{4/3} dx\right]^{1/4}.
\label{eq:e7}
\end{equation}
The constant 8.52 was substantially improved by E.H. Lieb and S. Oxford in 1981 \cite{LiOx81}, although the form 
of the  lower bound was slightly different,
\begin{equation}
E \geq -1.68\left\{\int\left[\sum_{i = 1}^N e_i \rho_i(x)\right]^{4/3} dx\right\}^{1/2} \left[\int \rho(x)^{4/3} dx\right]^{1/2}.
\label{eq:e8}
\end{equation}
However, in the more standard case, when $e_1 = \ldots = e_N = e$, the two inequalities are reduced to a similar form,
\begin{equation}
E \geq -Ce^{2/3}\int \rho(x)^{4/3} \, dx, 
\label{eq:e9}
\end{equation}
where $C$ is a positive constant. The best value for $C$ is unknown, but Lieb and Oxford \cite{LiOx81} 
proved that it is larger or equal than $1.234$. 
The constant they obtained in general, $1.68$, was found by first assigning to each particle a spherically symmetric charge distribution; however, they were not able to obtain the distribution that completely minimized the constant. Then Chan and Handy, in 1999 \cite{ChHa99}, found a better 
value, 1.636, by optimizing numerically this distribution of charge. It is this last constant, as far as we know, that is the smallest value for $C$ that has been found up--to--date. During the last thirty years, after the work of Lieb and Oxford \cite{LiOx81},  there has been a special interest in quantum chemistry in constructing corrections to the Lieb--Oxford term involving the gradient of the single particle density. This interest arises with the expectation that states with a relatively small kinetic energy have a smaller indirect part (see, e.g., \cite{LePe93,PeBuEr96,VeMeTr09} and references therein).

\bigskip
Here we will provide a lower bound that has a  smaller value for $C$ than any of the previously mentioned results, but at the price of obtaining another term that involves a gradient of a power of the single particle density. Our main result is the following:

\begin{theorem}\label{teorema}
 For any normalized wave function $\psi(x_1, \dots, x_N)$ and any $0 < \alpha < 1$ we have the estimate
\begin{eqnarray}\label{exch}
&  E(\psi) \ge -  {1.4508}{(1-\alpha)^{-1/3}}e^{2/3}\int_{\R^3} \rho(x)^{4/3} dx \nonumber \\
& - 0.2382(1-\alpha) ^{-1/6} e^{5/6}\left((\sqrt \rho, |p| \sqrt \rho)\right)^{1/2}  
\left(\int_{\R^3} \rho(x)^{4/3} dx \right)^{1/2} - e(1.1781/\alpha) 
(\sqrt \rho, |p| \sqrt \rho)
\label{eq:Main}
\end{eqnarray}
where
\begin{equation}
(\sqrt \rho, |p| \sqrt \rho) :=  \int_{\R^3} |\widehat{\sqrt \rho}(k)|^2 |2\pi k| d k = \frac{1}{2\pi^2} \int_{\R^3} \int_{\R^3} \frac{|\sqrt{\rho(x)} - \sqrt{\rho(y)}|^2 }{|x-y|^4} dx dy  \ ,
\label{eq:KE}
\end{equation}
where $\widehat f(k)$ denotes the Fourier-transform
$$
\widehat f (k) = \int_{\R^3} e^{-2\pi i k \cdot x} f(x) d x\ .
$$
\end{theorem}

\begin{remarks}

i) For many physical states, the contribution of the last two terms in (\ref{eq:Main}) is small compared with the contribution of the first term. We illustrate this fact in the Appendix; 

ii) For the second equality in (\ref{eq:KE})  see, e.g., \cite{LiLo01}, Section 7.12, equation (4), p. 184;

iii) It was already noticed by Lieb and Oxford (see the remark after equation (26), p. 261 on \cite{LiOx81}), that somehow for uniform densities the Lieb--Oxford constant should be $1.45$ instead of $1.68$; 

iv) In the same vein, J. P. Perdew \cite{Pe91}, by employing results for a uniform electron gas in its low density limit showed that in the Lieb--Oxford bound one ought to have $C \ge 1.43$ (see also, \cite{LePe93}).

\end{remarks}

\section{Proof of (\ref{exch})}

In this section we give the proof of the new lower bound on $E$. Through out this section we set  $e=1$. 
The proof of this theorem hinges on the following lemma, originally due to L. Onsager \cite{On39}. We quote it from \cite{LiSe09} for the case where
all the charges are equal.
\begin{lemma}[Onsager's lemma]
Consider $N$ unit point charges located at the distinct points $x_1, \dots, x_N$. For each $1 \le i \le N$ let $\mu_{x_i}$ be a nonnegative, bounded function
that is radially symmetric about he point $x_i$ and whose integral $\int \mu_{x_i}(x) dx = 1$. Then for any non-negative integrable function $\rho$ we have the inequality
\begin{equation}
\sum_{i < j} \frac{1}{|x_i-x_j|} \ge -D(\rho, \rho) + 2 \sum_{i=1}^N D(\rho, \mu_{x_i}) - \sum_{i=1}^N D(\mu_{x_i}, \mu_{x_i}) \ .
\end{equation}
\end{lemma}
For the simple proof we refer the reader to \cite{LiSe09}.
As in \cite{LiOx81} for each $1 \le i \le N$ we set
\begin{equation}
\mu_{x_i}(x) = \lambda^3 \rho_\psi(x_i) \mu(\lambda \rho(x_i)^{1/3}(x-x_i))
\end{equation}
where $\mu$ is a non-negative, bounded, radial function that integrates to one.
Applying Onsager's Lemma with $\rho = \rho_\psi$ to the indirect term 
\begin{equation}
E( \psi) = \langle \psi \sum_{i < j} \frac{1}{|x_i-x_j|} \psi \rangle - D(\rho_\psi , \rho_\psi)
\end{equation}
one finds after simple calculations
\begin{equation} \label{lowerbound}
 E(\psi) \ge - \int_{\R^3 \times \R^3} dy dx  \rho(y) \lambda \rho(x)^{4/3} R(\lambda \rho(x)^{1/3}(x-y)) - 
\lambda D(\mu,\mu) \int_{\R^3} \rho(x)^{4/3} dx
\end{equation}
Here
\begin{equation}
R(t) = \frac{1}{t} - \phi(t)
\end{equation}
where
\begin{equation}
\phi(t) = \int_{\R^3} \min \left(\frac{1}{t}, \frac{1}{|y|}\right) \mu(y) dy \ .
\end{equation}
We set
$$
F =  \int_{\R^3 \times \R^3} dy dx  \rho(y) \lambda \rho(x)^{4/3} R(\lambda \rho(x)^{1/3}(x-y)) ,
$$
$$
K =  \int_{\R^3 \times \R^3} dy dx \frac{  [\rho(y)^{1/2} - \rho(x)^{1/2}] ^2}{|x-y|^4} , 
$$
$$
L = \int_{\R^3} \rho(x)^{4/3} dx,
$$
$$
M_1 =   \frac{2 \pi }{3}\int_{|z|<1} |z|^2 \mu(z) dz,
$$
and
$$
M_2 = \left(\int_{\R^3} |z|^4 R(z)^2 d z\right)^{1/2} \ .
$$
With this notation we have:
\begin{lemma}
\begin{eqnarray}
 F \le \frac{M_1}{\lambda^2} L+ \frac{1}{\lambda^{3/2}} K^{1/2} \left[F^{1/2} +\frac{M_2}{\lambda} L^{1/2}\right] \ .
\end{eqnarray}
\end{lemma}
\begin{proof}
We write
\begin{eqnarray} \label{tobeestimated}
& \int_{\R^3 \times \R^3} dy dx  \rho(y) \lambda \rho(x)^{4/3} R(\lambda \rho(x)^{1/3}(x-y))
=  \int_{\R^3 \times \R^3} dy dx  \lambda \rho(x)^{7/3} R(\lambda \rho(x)^{1/3}(x-y))  \nonumber \\
&+
  \int_{\R^3 \times \R^3} dy dx  [\rho(y)^{1/2} - \rho(x)^{1/2}]  [\rho(y)^{1/2} + \rho(x)^{1/2}] \lambda \rho(x)^{4/3} R(\lambda \rho(x)^{1/3}(x-y))
 \end{eqnarray}
 The first term on the right side can be computed and yields
 \begin{equation}
\int_{\R^3 \times \R^3} dy dx  \lambda \rho(x)^{7/3} R(\lambda \rho(x)^{1/3}(x-y))
=\frac{1}{\lambda^2} \int_{\R^3} \rho(x)^{4/3} dx \int_{\R^3} R(|z|) dz
\end{equation} 
Now
$$
 \int_{\R^3} R(|z|) dz = 2 \pi  - \int_{|z| < 1} \phi(z) dz = 2\pi -  \int_{\R^3} \Phi(z) \mu(z) dz 
$$
where $\Phi$ is the potential of a uniform charge supported in the unit ball with total charge $4\pi/3$.
This function can be readily computed to be
$$
\Phi(z) = \begin{cases}  \frac{4 \pi}{3}  + \frac{4 \pi}{6} (1-|z|^2)  & {\rm if}\  |z| <1 \\   \frac{4 \pi}{3} \frac{1}{|z|} &{\rm  if} \ |z| \ge 1, \end{cases}
$$
and we get
$$
 \int_{\R^3} R(|z|) dz = \frac{2 \pi }{3}\int_{|z|<1} |z|^2 \mu(z) dz \ .
$$
Thus, we have that
\begin{equation} \label{basic}
\int_{\R^3 \times \R^3} dy dx  \lambda \rho(x)^{7/3} R(\lambda \rho(x)^{1/3}(x-y)) = 
 \frac{1}{\lambda^2} \int_{\R^3} \rho(x)^{4/3} dx   \frac{2 \pi }{3}\int_{|z|<1} |z|^2 \mu(z) dz
\end{equation}
The second term in (\ref{tobeestimated}) reads
\begin{eqnarray} \label{twoterms}
 \int_{\R^3 \times \R^3} dy dx  [\rho(y)^{1/2} - \rho(x)^{1/2}]  \rho(y)^{1/2} \lambda \rho(x)^{4/3} R(\lambda \rho(x)^{1/3}(x-y)) \nonumber \\ 
+  \int_{\R^3 \times \R^3} dy dx  [\rho(y)^{1/2} - \rho(x)^{1/2}]  \rho(x)^{1/2} \lambda \rho(x)^{4/3} R(\lambda \rho(x)^{1/3}(x-y))
\end{eqnarray}
 Further,
\begin{eqnarray}
& \int_{\R^3 \times \R^3 } dy dx  [\rho(y)^{1/2} - \rho(x)^{1/2}]  \rho(y)^{1/2} \lambda \rho(x)^{4/3} R(\lambda \rho(x)^{1/3}(x-y)) \nonumber \\ 
&=  \int_{\R^3 \times \R^3} dy dx \frac{  [\rho(y)^{1/2} - \rho(x)^{1/2}] }{|x-y|^2}  \rho(y)^{1/2}|x-y|^2  \lambda \rho(x)^{4/3} R(\lambda \rho(x)^{1/3}(x-y))
\end{eqnarray}
and by Schwarz's inequality this is bounded above by
\begin{eqnarray}
&  \left[ \int_{\R^3 \times \R^3} dy dx \frac{  [\rho(y)^{1/2} - \rho(x)^{1/2}] ^2}{|x-y|^4}  
\right]^{1/2}  \nonumber \\
& \times    \left[ \int_{\R^3 \times \R^3} dy dx  \rho(y) \lambda^2  \rho(x)^{8/3} |x-y|^4R(\lambda \rho(x)^{1/3}(x-y))^2 \right]^{1/2}
\end{eqnarray}
Here we note that the function $R(t)=0$ for $t >1$ and hence we can assume that the domain of integration is such
that $\lambda \rho(x)^{1/3} |x-y| < 1$. Moreover,
$$
R(t) < \frac{1}{t} \ .
$$
Thus, the second factor is bounded by
$$
\frac{1}{\lambda^{3/2}}\left[\int_{\R^3 \times \R^3} dy dx  \rho(y)  \lambda \rho(x)^{4/3} R(\lambda \rho(x)^{1/3}(x-y)) \right]^{1/2}
$$
i.e., 
\begin{eqnarray}
& \int_{\R^3 \times \R^3} dy dx  [\rho(y)^{1/2} - \rho(x)^{1/2}]  \rho(y)^{1/2} \lambda \rho(x)^{4/3} R(\lambda \rho(x)^{1/3}(x-y)) \nonumber \\ 
& \le   \frac{1}{\lambda^{3/2}} \left[ \int_{\R^3 \times \R^3} dy dx \frac{  [\rho(y)^{1/2} - \rho(x)^{1/2}] ^2}{|x-y|^4}  \right]^{1/2}
\left[\int_{\R^3 \times \R^3} dy dx  \rho(y) \lambda   \rho(x)^{4/3} R(\lambda \rho(x)^{1/3}(x-y)) \right]^{1/2}
\end{eqnarray}
We estimate the  second term in (\ref{twoterms}) in a similar fashion. Schwarz's inequality leads to
\begin{eqnarray}
&\int_{\R^3 \times \R^3} dy dx  [\rho(y)^{1/2} - \rho(x)^{1/2}]  \rho(x)^{1/2} \lambda \rho(x)^{4/3} R(\lambda \rho(x)^{1/3}(x-y)) \nonumber \\
& \le \left[ \int_{\R^3 \times \R^3} dy dx \frac{  [\rho(y)^{1/2} - \rho(x)^{1/2}] ^2}{|x-y|^4}  
\right]^{1/2}  \nonumber \\
& \times \left[ \int_{\R^3 \times \R^3} dy dx  |x-y|^4  \lambda^2 \rho(x)^{11/3} R(\lambda \rho(x)^{1/3}(x-y))^2 \right]^{1/2} \ .
\end{eqnarray}
Changing variables $ y \to z = \lambda \rho(x)^{1/3}(x-y)$ yields for the second factor
\begin{equation}
\lambda^{-5/2} \left[\int_{\R^3} \rho(x)^{4/3} dx \int_{\R^3} |z|^4 R(z)^2 d z\right]^{1/2} \ .
\end{equation}
\end{proof}
\begin{proof}[Proof of Theorem \ref{teorema}]
For any $\alpha > 0$ we have that
$$
\frac{1}{\lambda^{3/2}} K^{1/2} F^{1/2}  \le \frac{K}{4 \alpha \lambda^3} +\alpha F,
$$
and hence,
$$
F \le \frac{M_1}{(1-\alpha) \lambda^2} L+ \frac{M_2}{(1-\alpha) \lambda^{5/2}} K^{1/2} L^{1/2} +\frac{K}{4(1-\alpha) \alpha \lambda^3} \ .
$$
Thus, the absolute value of the indirect term is bounded above by
\begin{equation}
|E|  \le  \left[\frac{1}{(1-\alpha) \lambda^2}  M_1  + \lambda D(\mu,\mu)\right]
L + \frac{M_2}{(1-\alpha) \lambda^{5/2}} K^{1/2} L^{1/2} +\frac{1}{4(1-\alpha) \alpha \lambda^3}K \ .
\end{equation}
Now we optimize the {\it first} term with respect to $\lambda$ and obtain the upper bound
\begin{equation}
\frac{3M_1^{1/3} D(\mu,\mu)^{2/3}}{2^{2/3} (1-\alpha)^{1/3}}    L 
+  \frac{1}{(1-\alpha)^{1/6} } \frac{M_2 D(\mu,\mu)^{5/6}} {(2M_1)^{5/6}}  K^{1/2} L^{1/2}
+\frac{1}{4 \alpha} \left( \frac{D(\mu,\mu)} {2M_1}\right) K \ .
\end{equation}
If we optimize the expression
$$
D(\mu,\mu)^{2/3}M_1^{1/3} = D(\mu,\mu)^{2/3} \left( \frac{2 \pi }{3}\int_{|z|<1} |z|^2 \mu(z) dz\right)^{1/3}
$$
with respect to the function $\mu$ we find that the optimizing $\mu$ is the uniform distribution on the unit ball.
Simple calculations yield for this distribution
$$
D(\mu,\mu) = \frac{3}{5},  \qquad \frac{2 \pi }{3}\int_{|z|<1} |z|^2 \mu(z) dz = \frac{2 \pi}{5} \ ,
$$
and moreover $M_2 = \sqrt{{23\pi}/{2310}}$.
Thus,
\begin{equation} E \ge - \frac{9}{10} \left(\frac{4\pi}{3(1-\alpha)}\right)^{1/3} L 
-  \frac{1} {(1-\alpha)^{1/6} } (\frac{3}{4\pi})^{5/6}  \sqrt{\frac{23\pi}{2310}}K^{1/2} L ^{1/2} - \frac{3}{16 \pi \alpha} K  \ .
\end{equation}
\end{proof}

\section{Appendix}

In this Appendix we illustrate how for some states, the first term in (\ref{exch}) can be 
much larger than the second. In our first example we will consider a single particle density given by an exponential, normalized to $N$ (the number of particles). We show that in this case, the quotient between the first and the last term goes like $1/N^{1/3}$, so that in the large particle number limit, the last two terms in (\ref{exch}) are negligible. 

A function $f: \mathbb{R}^3 \to \mathbb{C}$ is said to be in $H^{1/2}(\mathbb{R}^3)$ if $f \in L^2(\mathbb{R}^3)$ and $\left|k\right|^{1/2}\widehat{f}(k)$ is in $L^2(\mathbb{R}^3)$. For functions in this set, the following identity holds,
\begin{equation}
K(f) \equiv \frac{1}{2\pi^2}\iint\frac{\left|f(x) - f(y)\right|^2}{\left|x - y\right|^4}\,dx\,dy = \int\lvert 2\pi k \rvert\lvert\widehat{f}(k)\rvert^2\,dk,
\end{equation}
(see, e.g., \cite{LiLo01}, Section 7.12, equation (4), p. 184).

It suffices to prove that the quotient
\begin{equation}
\frac{K(\sqrt{\rho})}{\int\rho^{4/3}\,dx} \equiv Q(\rho)
\end{equation}
can be made arbitrarilly small. To this end, we consider the function,
\begin{equation}
\rho(x) = a e^{-b\left|x\right|}.
\end{equation}
Since $\int\rho(x)\,dx = N$,
\begin{equation}
b = \left(\frac{8\pi a}{N}\right)^{1/3}.
\end{equation}
Certainly $\sqrt{\rho} \in L^2(\mathbb{R}^3)$, and
\begin{equation}
\widehat{\sqrt{\rho}}(k) = \frac{32\pi a^{1/2}b}{(\left|4\pi k\right|^2 + b^2)^2},
\end{equation}
so $\sqrt{\rho} \in H^{1/2}(\mathbb{R}^3)$. Therefore,
\begin{equation}
K(\sqrt{\rho}) = \frac{8a}{3b^2}.
\end{equation}
On the other hand,
\begin{equation}
\int\rho^{4/3}\,dx = \frac{27a^{4/3}\pi}{8b^3}.
\end{equation}
So we conclude that
\begin{equation}
Q(\rho) = \frac{2^{7}}{3^4 \pi^{2/3} N^{1/3}},
\end{equation}
and taking the limit $N \to \infty$, the desired result is obtained.

As a second example, which is more relevant from the physical point of view, 
we use the Thomas--Fermi density of an atom, $\rho_{TF}$. It is well known 
that the Thomas--Fermi density of a neutral  atom of nuclear charge $Z$ satisfies the following scaling, 
\begin{equation}
\rho_{TF}(x)=Z^2f(Z^{1/3} x),
\label{eq:A1}
\end{equation}
where $f$ denotes the Thomas--Fermi density in the case $Z=1$. Using this simple scaling relation, one can immediately check that 
$$
\int \rho_{TF}^4/3 \, dx = c Z^{5/3}, 
$$
for some positive constant $c$, independent of $Z$. On the other hand, also using this simple scaling relation one sees that
$$
K(\sqrt{\rho_{TF}}) = d \, Z^{4/3},
$$
for some positive constant $d$ (independent of $Z$). Thus, again we observe the same dependence in the number of particles as in the previous example, i.e., 
$$
Q(\rho_{TF}) = \frac{\tilde c}{Z^{1/3}}.
$$
Thus, for large values of $Z$, the second and third terms in (\ref{exch}) are negligible.

\end{document}